\documentclass[12pt]{amsart}
\usepackage[margin=1.25in]{geometry}
\usepackage[english]{babel}
\usepackage[utf8]{inputenc}
\usepackage{subcaption}
\usepackage{amsmath}
\usepackage{amssymb}
\usepackage{amsfonts}
\usepackage{amsthm}
\usepackage{mathrsfs}
\usepackage[all]{xy}
\usepackage[pdftex]{graphicx}
\usepackage{color}
\usepackage{cite}
\usepackage{url}
\usepackage{comment}
\usepackage{indent first}
\usepackage[labelfont=bf,labelsep=period,justification=raggedright]{caption}
\usepackage[english]{babel}
\usepackage[utf8]{inputenc}
\usepackage[colorlinks=true, citecolor=cyan, linkcolor=magenta]{hyperref}
\usepackage[colorinlistoftodos]{todonotes}
\usepackage{tkz-fct}
\usepackage{tikz}
\usepackage{color, colortbl}
\usepackage{mathtools}
\usepackage{microtype}
\usepackage{cleveref}
\usepackage{enumitem}

\usepackage[ruled, noline, noend]{algorithm2e}
\LinesNumbered
\DontPrintSemicolon
\SetKwInput{KwConstant}{Constants}
\SetKwInput{KwInput}{Input}
\SetKwInOut{KwOutput}{Output}
\SetKwProg{Procedure}{Procedure}{}{}
\SetKwComment{Comment}{/* }{ */}

\usetikzlibrary{calc}

\newcommand{\eps}{\varepsilon}

\theoremstyle{plain}
\newtheorem{thm}{Theorem}
\newtheorem{lemma}[thm]{Lemma}

\theoremstyle{definition}
\newtheorem{defn}[thm]{Definition}

\theoremstyle{remark}

\numberwithin{equation}{section}
\numberwithin{thm}{section}

\title[]{Improving the runtime of algorithmic polarization of hidden markov models}

\author[V. Bian]{Vincent Bian}
\address{V. Bian: Department of Physics, Massachusetts Institute of Technology, Cambridge, MA 02139}
\email{vinvinb@mit.edu}

\author[R. Madhukara]{Rachana Madhukara}
\address{R. Madhukara: Department of Mathematics, Massachusetts Institute of Technology, Cambridge, MA 02139}
\email{rachanam@mit.edu}

\begin{document}

\maketitle

\begin{abstract}
    We improve the runtime of the linear compression scheme for hidden Markov sources presented in a 2018 paper of Guruswami, Nakkiran, and Sudan. Under the previous scheme, compressing a message of length $n$ takes $O(n \log n)$ runtime, and decompressing takes $O(n^{1 + \delta})$ runtime for any fixed $\delta > 0.$ We present how to improve the runtime of the decoding scheme to $O(n \log n)$ by caching intermediate results to avoid repeating computation.
\end{abstract}

\section{Introduction}


The problem of designing coding schemes for stochastic channels is hard. In a recent paper of Guruswami, Nakkiran, and Sudan \cite{hmm_polar}, they consider hidden Markov channels, where the states evolve according to some Markov process. The ``hidden" refers to the fact that one cannot determine the state from the output of the channel. One important thing to note here is that error introduced by the channel is only dependent on the state of the Markov process.

In designing encoding and decoding algorithms for the problem of hidden Markov channels, Markovian sources become quite important due to a well-known correspondence relating error-correction for additive Markov chains to compression and decompression algorithms for Markovian sources (when the compression is linear). Recall that additive channels are the those which map inputs from some alphabet $\Sigma$ to outputs over $\Sigma$, where there is an abelian group defined on $\Sigma$ (needed for the additive purposes). The channel creates an error sequence from the associated Markov source which is independent of the input sequence. Then the output of the channel is the coordinatewise sum of the input sequence with the error sequence. Note that in \cite{hmm_polar}, the alphabet $\Sigma$ is taken to be a finite field of prime cardinality. More formally, 

\begin{defn}
An \textit{additive Markov channel} $\mathcal{C}_{\mathcal{H}}$, specified by a Markov source $\mathcal{H}$ over alphabet $\mathbb{F}_q$, is a randomized map $\mathcal{C}_{\mathcal{H}} : \mathbb{F}^*_q \rightarrow \mathbb{F}^*_q$ obtained as follows: On channel input $X_1,\dots,X_n$, the channel outputs $Y_1,\dots,Y_n$ where $Y_i=X_1+Z_i$ where $Z = (Z_1,\dots,Z_n) \sim \mathcal{H}_n$.  
\end{defn}

Now we define a Markovian source.

\begin{defn}
A (hidden) Markov source over some alphabet $\Sigma$ is a Markov chain on a finite state space where each state $s$ has some distribution associated to it. More formally, the sequence $Z_1, \dots, Z_n$ is a Markov source if and only if 
\begin{enumerate}
    \item there is some $k$-state time invariant Markov sequence $X_0, X_1, \dots, X_n, \dots$ which are generated by the transition matrix $M$;
    \item there are some $k$ distributions $P^{(1)}, \dots, P^{(k)} \in \Delta(\Sigma)$ such that $Z_i \sim P^{(X_i)}$. 
\end{enumerate}
\end{defn}

In \cite{hmm_polar}, a construction of codes for additive Markov channels is given. In particular, these constructions get $\eps$ close to capacity when compressing $n$ bits, where $n$ is polynomial in $1/\eps$ and in the mixing time of the Markov chain. Additionally, they show that one can compress $n$ bits to its entropy up to an additive factor of $\eps n$, where $n$ is once again polynomial in $1/\eps$. While the paper only considers additive channels, the authors (and we) believe that the results in the paper should extend to more general symmetric channels with careful bookkeeping and the consideration of conditional probabilities.  

In this paper, we discuss these results and summarize the main ideas from the given compression and decompression algorithms. Additionally, we improve the \textsc{Polar-Decompress} algorithm runtime to $O(n\log n)$.  

\section{Preliminaries}
\subsection{Notation}
We denote our polarization/mixing matrix as $M$ (of size $k \times k$), and we will consider applying $M$ to an input $t$ times, so the result of polarization is $P_m(Z) = M^{\otimes t} Z.$ We consider our alphabet $\Sigma$ to be $\mathbb{F}_q$ for some prime $q,$ so $M \in \mathbb{F}^{k \times k}_q.$ Additionally, the Markov source is denoted as $\mathcal{H}$. 

\subsection{Forward Algorithm}

The \textit{Forward Algorithm} is prevalent in the preprocessing, compression and decompression algorithms. Therefore, we provide an overview of what the algorithm does and how it works in this subsection.

Given a Markov source $\mathcal{H}$ with $\ell$ states, let the first $j$ samples be $Y_1, \ldots, Y_j \sim \mathcal{H}_n.$ For a setting $y_1, \ldots, y_{j - 1},$ the \emph{Forward Algorithm} computes the distribution of $Y_n$ given that $(Y_1, \ldots, Y_{j - 1}) = (y_1, \ldots, y_{j - 1}).$ The details of how it is implemented can be found under Algorithm A.1 of \cite{hmm_polar}, but the main idea is to successively compute the distribution of the underlying state of $\mathcal{H}$ after each sample via Bayesian updates. Each update takes time $O(\ell^2)$ (since there are roughly this many possible transitions to consider), and we must do $j$ of these updates, for a total runtime of $O(j \ell^2).$ This is implemented using dynamic programming. 

\section{Polar-Decompress Speed Up Process}

The following theorem is largely identical to Theorem~2.9 from \cite{hmm_polar}, except $\textsc{Polar-Decompress}$ is replaced with $\textsc{Fast-Polar-Decompress}$ and its runtime is reduced from $O(n^{3/2}\ell^2 + n\log n)$ to $O(n \log n).$

\begin{thm}
\label{thm:fast_polar_decompress}
For every prime $q$ and mixing matrix $M \in \mathbb{F}_{q}^{k\times k}$ there exists a preprocessing algorithm (\textsc{Polar-Preprocess}, Algorithm 6.3 in \cite{hmm_polar}), a compression algorithm (\textsc{Polar-Compress}, Algorithm 4.1 in \cite{hmm_polar}), a decompression algorithm (\textsc{Fast-Polar-Decompress}, Algorithm \ref{alg:fast_pol_dec}) and a polynomial $p(\cdot)$ such that for every $\epsilon > 0$, the following property holds:
\begin{enumerate}
    \item \textsc{Polar-Preprocess} is a randomized algorithm that takes as input a Markov source $\mathcal{H}$ with $\ell$ states, and $t \in \mathbb{N},$ and runs in time $\mathrm{poly}(n, \ell, 1/\varepsilon, q)$ where $n = k^{2t}$ and outputs auxiliary information for the compressor and decompressor (for $\mathcal{H}_n$).
    \item \textsc{Polar-compress} takes as input a sequence $Z \in \mathbb{F}_q^n$ as well as the auxiliary information output by the preprocessor, runs in time $\mathcal{O} (n\log n)$, and outputs a compressed string $\Tilde{U} \in \mathbb{F}_q^{\overline{H}(Z)+\epsilon n}$. Further, for every auxiliary input, the map $Z \to \Tilde{U}$ is a linear map. 
    \item \textsc{Fast-Polar-Decompress} takes as input a Markov source $\mathcal{H}$, a compressed string $\Tilde{U} \in \mathbb{F}_q^{\overline{H}(Z)+\epsilon n}$ and the auxiliary information output by the preprocessor, runs in time $\mathcal{O}(n\log n)$ and outputs $\hat{Z} \in \mathbb{F}_q^n.$ 
\end{enumerate}

The guarantee provided by the above algorithms is that with probability at least $1 - \exp(-\Omega(n)),$ the
Preprocessing Algorithm outputs auxiliary information $S$ such that
$$\Pr_{Z \sim \mathcal{H}_n} [\textsc{Polar-Decompress}(\mathcal{H}, S; \textsc{Polar-Compress}(Z; S)) \neq Z] \leq \mathcal{O}(\frac{1}{n^2}),$$
provided $n > p(\tau/\varepsilon)$ where $\tau$ is the mixing time of $\mathcal{H}.$
\end{thm}

Please refer to the Appendix for a complete description of \textsc{Polar-Decompress} and the Forward Algorithm. Additionally, note that \textsc{Polar-decompress} makes black-box use of the \textsc{Fast-decoder} algorithm from \cite[Algorithm 4]{strong_polar}.

We now present the optimized version of \textsc{Polar-Decompress}, which effectively caches the intermediate values computed in \textsc{ForwardInfer}. Note that \textsc{Polar-Decompress} and \textsc{ForwardInfer} are written with different notation, and we will use the notation from \textsc{Polar-Decompress}.

\begin{algorithm}
\caption{\textsc{Fast-Polar-Decompress}} \label{alg:fast_pol_dec}
\KwConstant{$M \in \mathbb{F}^{k \times k}_q,$ $m = k^t,$ $n = m^2$}
\KwInput{Markov source $\mathcal{H}$ with state space $[\ell]$ and stationary distribution $\pi,$ and compressed strings $U^1_{S_1}, U^2_{S_2}, \ldots, U^m_{S_m} \in \mathbb{F}^{m}_q$}
\KwOutput{$\hat{Z} \in \mathbb{F}^{m \times m}_q$}
\Procedure{$\textsc{Fast-Polar-Decompress}(\mathcal{H}; U^1_{S_1}, U^2_{S_2}, \ldots, U^m_{S_m})$}{
  \ForEach{$i \in [m]$}{
    Initialize $v_{i, 0} \gets \pi$ \;
  }
  \ForEach{$j \in [m]$}{
    \eIf{$j \leq (1 - \varepsilon)m$}{
      \ForEach{$i \in [m]$}{
        Compute $\mathcal{D}_{z^j_i | z^{<j}_i} := \mathbb{E}_{z \sim \Pi v_{i, j - 1}}[\mathcal{S}_z]$ \;
      }
      
      Define $U^j \in (\mathbb{F}_q \cup \{\bot\})^m$ by filling unspecified entries of $U^j_{S_j}$ with $\bot$ \;
      Set $\hat{Z}^j \gets \textsc{Fast-Decoder}(\mathcal{D}_{z^j | z^{<j}} ; U^j)$ \;
      
      \ForEach{$i \in [m]$}{
        Define $v_{i, j} \in \Delta([\ell])$ by $v_{i, j}(z) \gets \dfrac{(\Pi v_{i, j - 1})_z \cdot \mathcal{S}_z(\hat{Z}_i^j)}{\sum_{y \in [\ell]} (\Pi v_{i, j - 1})_y \cdot \mathcal{S}_y(\hat{Z}_i^j)},$ treating $v_{i, j - 1}$ as a vector in $\mathbb{R}^{\ell}$ \;
      }
    }{
      Set $\hat{Z}^j \gets (M^{-1})^{\otimes t} U^j_{S_j}$
    }
  }
  \Return $\hat{Z}$
}
\end{algorithm}

Now, we prove Theorem~\ref{thm:fast_polar_decompress}.
\begin{proof}
    The proofs of properties (1) and (2) are the same as in \cite{hmm_polar}. To show property (3), it suffices to show that $\textsc{Fast-Polar-Decompress}$ runs in time $O(n \log n).$ To prove the correctness claim at the end, it suffices to show $\textsc{Fast-Polar-Decompress}$ produces the same output as $\textsc{Polar-Decompress}.$ On each iteration of the main for loop, both versions of the algorithm set $\hat{Z}^j$ to $\textsc{Fast-Decoder}(\mathcal{D}_{z^j | z^{<j}} ; U^j),$ so it suffices to show that both algorithms compute the same value for $\mathcal{D}_{z^j_i | z^{<j}_i},$ for all $i$ and $j.$

    We prove the following two statements by induction on $j,$ for $j \leq (1 - \varepsilon)m$:
    \begin{enumerate}
        \item When the subroutine $\textsc{ForwardInfer}(\mathcal{H}; 1, \hat{Z}^{<j}_i)$ is called in line 4 of $\textsc{Polar-Decompress},$ the value of $s_{j - 1}$ computed in $\textsc{ForwardInfer}$ is equal to the value of $r_{i, j - 1}$ computed in $\textsc{Fast-Polar-Decompress}$
        \item Both $\textsc{Polar-Decompress}$ and $\textsc{Fast-Polar-Decompress}$ compute the same value for $\mathcal{D}_{z^j_i | z^{<j}_i}$
    \end{enumerate}
    
    For the base case $j = 1,$ $\textsc{Polar-Decompress}$ sets $\mathcal{D}_{z^j_i | z^{<j}_i}$ to $\textsc{ForwardInfer}(\mathcal{H}; 1, \hat{Z}^{<1}_i).$ When $\textsc{ForwardInfer}$ is called, it sets $s_0$ to the stationary distribution $\pi,$ which is also the value that $v_{i, 0}$ gets initialized to. This proves the first half of the base case. After $s_0$ is initialized, $\textsc{ForwardInfer}$ skips the for loop entirely, sets $s_1$ to $\Pi s_0 = \pi,$ and returns $\mathbb{E}_{i\sim s_1}[\mathcal{S}_i] = \mathbb{E}_{i \sim \pi}[\mathcal{S}_i].$

    Meanwhile, when $j = 1,$ $\textsc{Fast-Polar-Decompress}$ directly sets $\mathcal{D}_{z^j_i | z^{<j}_i}$ to $\mathbb{E}_{z \sim \Pi v_{i, j - 1}}[\mathcal{S}_z].$ Since $v_{i, j - 1} = v_{i, 0} = \pi,$ and $\Pi \pi = \pi,$ we have $\mathcal{D}_{z^j_i | z^{<j}_i} = \mathbb{E}_{z \sim \pi}[\mathcal{S}_z],$ the same value as in $\textsc{Polar-Decompress}.$ This completes the base case.

    In $\textsc{Polar-Decompress},$ the entries of $\hat{Z}$ are never changed after they are set. Consider the computation in the for loop in $\textsc{ForwardInfer}.$ Each value of $s_t$ only depends on $s_{t - 1},$ and the value of $y_t.$ When $\textsc{ForwardInfer}$ is called, $y$ is set to $\hat{Z}^{<j}.$ Thus, the inputs $y_1, y_2, \ldots, y_{j - 2}$ when $\textsc{ForwardInfer}(\mathcal{H}; j - 1, \hat{Z}^{<j}_i)$ is called will be the same as when $\textsc{ForwardInfer}(\mathcal{H}; j - 1, \hat{Z}^{<j}_i)$ is called in the previous iteration. This means that values of $s_0$ through $s_{j - 2}$ will be the same as well. By the inductive hypothesis, this means $s_{j - 2} = r_{i, j - 2}.$

    Thus, $s_{j - 1}$ is computed using
    $$s_{j - 1}(z) \gets \dfrac{(\Pi s_{j - 2})_z \cdot \mathcal{S}_z(y_t)}{\sum_{j \in [\ell]} (\Pi s_{j - 2})_j \cdot \mathcal{S}_j(y_t)} = \dfrac{(\Pi s_{j - 2})_z \cdot \mathcal{S}_z(\hat{Z}^{j - 1}_i)}{\sum_{j \in [\ell]} (\Pi s_{j - 2})_j \cdot \mathcal{S}_j(\hat{Z}^{j - 1}_i)},$$ where we have substituted in the value of the input $y_t.$ This is exactly the same expression as the computation for $v_{i, j - 1}(z)$ in line 11 of $\textsc{Fast-Polar-Decompress}$, just with some of the indices named differently. This proves the first statement in the inductive hypothesis.

    To show the second statement of the inductive hypothesis, note that $\textsc{ForwardInfer}(\mathcal{H}; j - 1, \hat{Z}^{<j}_i)$ outputs $\mathbb{E}_{z \sim s_j}[\mathcal{S}_z],$ and $s_j := \Pi s_{j - 1} = \Pi v_{i, j - 1}.$ However, $\mathbb{E}_{z \sim \Pi v_{i, j - 1}}[\mathcal{S}_z]$ is exactly the value $\textsc{Fast-Polar-Decompress}$ assigns to $\mathcal{D}_{z^j_i | z^{<j}_i}$ in line 7. This completes the inductive step. Hence, $\textsc{Fast-Polar-Decompress}$ and $\textsc{Polar-Decompress}$ produce the same output.
\end{proof}

Finally, we show that the runtime of $\textsc{Fast-Polar-Decompress}$ is indeed faster than $\textsc{Polar-Decompress}$.  

\begin{lemma}
The runtime of $\textsc{Fast-Polar-Decompress}$ is $O(n \log n)$.  
\end{lemma}

\begin{proof}
Line 8 takes $O(m),$ and line 9 takes $O(m \log m)$ \cite{strong_polar}. The runtimes of lines 7 and 11 are both proportional to $\ell^2,$ which is constant with respect to $m.$ Thus, the inner for loops both contribute $O(m)$ each time they are run. Since lines 6 through 11 are run $(1 - \varepsilon) m$ times, the total runtime is $O(m^2 \log m) = O(n \log n),$ as desired.
\end{proof}

\section{Acknowledgements}

This work began as a final project in a Harvard class on Information Theory in Computer Science (CS229r, Fall 2022). We thank Professor Madhu Sudan for his guidance on this project. We also thank Chi-Ning Chou and R. Emin Berker for helpful suggestions.

\nocite{*}
\bibliographystyle{alpha}
\bibliography{biblio.bib}

\newpage

\appendix

\section{Polar Decompress}

\begin{algorithm}
\caption{\textsc{Polar-Decompress} (Algorithm 4.2 from \cite{hmm_polar})} \label{alg:pol_dec}
\KwConstant{$M \in \mathbb{F}^{k \times k}_q,$ $m = k^t,$ $n = m^2$}
\KwInput{Markov source $\mathcal{H}$ and $U^1_{S_1}, U^2_{S_2}, \ldots, U^m_{S_m} \in \mathbb{F}^{m}_q$}
\KwOutput{$\hat{Z} \in \mathbb{F}^{m \times m}_q$}
\Procedure{$\textsc{Polar-Decompress}(\mathcal{H}; U^1_{S_1}, U^2_{S_2}, \ldots, U^m_{S_m})$}{
  \ForEach{$j \in [m]$}{
    \eIf{$j \leq (1 - \varepsilon)m$}{
      Compute the conditional distribution $\mathcal{D}_{z^j | z^{<j}}$ of $\overline{Z}^j$ given $\overline{Z}^{<j} = \hat{Z}^{<j},$ using the Forward Algorithm on the Markov source $\mathcal{H}$ \;
      Define $U^j \in (\mathbb{F}_q \cup \{\bot\})^m$ by filling unspecified entries of $U^j_{S_j}$ with $\bot$ \;
      Set $\hat{Z}^j \gets \textsc{Fast-Decoder}(\mathcal{D}_{z^j | z^{<j}} ; U^j)$ \;
    }{
      Set $\hat{Z}^j \gets (M^{-1})^{\otimes t} U^j_{S_j}$
    }
  }
  \Return $\hat{Z}$
}
\end{algorithm}

\section{Forward Algorithm}

\begin{algorithm}
\caption{Forward Algorithm (Algorithm A.1 from \cite{hmm_polar})} \label{alg:forward}
\KwInput{$n \in \mathbb{N}.$ Markov source $\mathcal{H}$ with state space $[\ell],$ alphabet $\Sigma,$ stationary distribution $\pi \in \Delta([\ell]),$ transition matrix $\Pi \in \mathbb{R}^{\ell \times \ell},$ and output distributions $\{\mathcal{S}_i \in \Delta(\Sigma)\}_{i \in [\ell]}.$ And $y = (y_1, y_2, \ldots, y_{n - 1})$ for $y_i \in \Sigma.$}
\KwOutput{Distribution $Y_n \in \Delta(\Sigma)$}
\Procedure{$\textsc{ForwardInfer}(\mathcal{H} = (\ell, \Sigma, \pi, \Pi, \{\mathcal{S}_i\}); n, y)$}{
  $s_0 \gets \pi$ \;
  \ForEach{$t = 1, 2, \ldots, n - 1$}{
    Define $s_t \in \Delta([\ell])$ by $s_t(i) \gets \dfrac{(\Pi s_{t - 1})_i \cdot \mathcal{S}_i(y_t)}{\sum_{j \in [\ell]} (\Pi s_{t - 1})_j \cdot \mathcal{S}_j(y_t)},$ treating $s_{t - 1}$ as a vector in $\mathbb{R}^{\ell}.$ \;
  }
  $s_n \gets \Pi s_{n - 1}$ \;
  \Return The distribution $Y_n := \mathbb{E}_{i \sim s_n}[\mathcal{S}_i]$
}
\end{algorithm}

\end{document}